\documentclass[fontsize=11pt, headings=normal, headinclude=true, paper=a4, pagesize, cleardoublepage=current, BCOR=10mm, fleqn, numbers=noenddot]{scrartcl}
\pdfoutput=1
\usepackage{lmodern}

\setkomafont{sectioning}{\normalcolor\bfseries}
\recalctypearea 
\setcapindent{0pt}

  \usepackage[english]{babel}  
  \usepackage[numbers,square]{natbib}
  \usepackage{mathrsfs}
  \usepackage{amsmath}
  \usepackage{amssymb}
  \usepackage{amsthm}
  \usepackage{enumitem}
  \usepackage[utf8x]{inputenc}    
  \usepackage[T1]{fontenc}         
  \usepackage{graphicx}            
  \usepackage{color}		
  \usepackage[english=british]{csquotes} 

\newcommand{\R}{\mathbb{R}}

\newcommand{\N}{\mathbb{N}}

\newcommand{\ud}{\mathrm{d}}
\newcommand{\ui}{\mathrm{i}}
\newcommand{\ue}{\mathrm{e}}

\renewcommand{\phi}{\varphi}

\newcommand{\eps}{\varepsilon}

\theoremstyle{theorem}
\newtheorem{thm}{Theorem}
\newtheorem{lemma}[thm]{Lemma}
\newtheorem{prop}[thm]{Proposition}
\newtheorem{corol}[thm]{Corollary}

\newtheorem*{thm*}{Theorem}
\newtheorem*{rem*}{Remark}

\theoremstyle{definition}
\newtheorem{definition}{Definition}

\title{A remark on the attainable set of the Schrödinger equation}
\author{Jonas Lampart 
\thanks{CNRS \& Laboratoire interdisciplinaire Carnot de Bourgogne (UMR 6303)
, Université de Bourgogne 
Franche-Comté, 9 Av. A. Savary, 21078 Dijon Cedex, France.
\texttt{jonas.lampart@u-bourgogne.fr}}}

\begin{document}
\maketitle

\begin{abstract}
 We discuss the set of wavefunctions $\psi_V(t)$ that can be obtained from a given initial condition $\psi_0$ by applying the flow of the Schrödinger operator $-\Delta + V(t,x)$ and varying the potential $V(t,x)$. We show that this set has empty interior, both as a subset of the sphere in $L^2(\R^d)$ and as a set of trajectories.
\end{abstract}

\section{Introduction}

In this letter we study the total set of wavefunctions $\psi_V(t)$ that can be obtained from a \emph{fixed} initial condition $\psi_0$ by applying the flow of the Schrödinger operator $-\Delta + V(t,x)$, where $t>0$ is \emph{any} time and $V$ is \emph{any} potential in $L^p_{\mathrm{loc}}(\R, L^\infty(\R^d))$, with $p>1$.

A famous result in this direction, due to Ball, Marsden and Slemrod~\cite{ball1982}, concerns the general evolution equation
\begin{equation}
 \left\{
 \begin{aligned}
 \partial_t \psi_u(t) &= A \psi_u(t) + u(t)B\psi_u(t)        \\
 \psi_u(0)&=\psi_0,
 \end{aligned}
 \right.\label{eq:evolution}
\end{equation}
where $(A,D(A))$ is a generator of a strongly continuous semi-group on some Banach space $X$, $B$ is a \emph{fixed} bounded operator, and $u(t)\in L^p_\mathrm{loc}([0,\infty))$, $p> 1$.
They proved that the set 
\begin{equation}
\mathcal{A}_B = \bigcup_{p>1} \left\{ \psi_u(t) \Big| t \geq 0,\, \psi_u \text{ solves~\eqref{eq:evolution} with } u\in L^p_\mathrm{loc}([0,\infty)) \right\}\label{eq:Att gen}
\end{equation}
is a countable union of compact subsets of $X$. If $X$ has infinite dimension, this means that it has empty interior, by Baire's theorem. This has important consequences for the problem of controlling such a system, since it implies that for a dense set of elements $\psi_1\in X$ there exist no $t\geq 0$ and $u \in L^p_\mathrm{loc}$ such that $\psi_u(t)=\psi_1$, i.e., given $B$ the system cannot be driven from $\psi_0$ to $\psi_1$ by choosing  $u$. 

Taking $X=L^2(\R^d)$, $A=\ui \Delta$, $D(A)=H^2(\R^d)$ and $(B\psi)(x)= -\ui V(x)\psi(x)$ with $V\in L^\infty(\R^d,\R)$ gives a result for the Schrödinger equation.
Of course, in this case, the attainable set $\mathcal{A}_V$ is always contained in the sphere of radius $\|\psi_0 \|$, which has empty interior in $L^2(\R^d)$. However, $\mathcal{A}_V$ also has empty interior in the relative topology of the sphere, see Turinici~\cite{turinici2000}. More recently, these results were generalised to include the case $p=1$ and to Radon measures by Boussa\"{\i}d, Caponigro and Chambrion~\cite{boussaid2014,boussaid2017}, and to non-linear equations by Chambrion and Thomann~\cite{chambrion2019, chambrion2020}.

All of these results are based on the observation that the solution operator to~\eqref{eq:evolution} can be expressed as a series of integral operators. From this, one deduces that the map $u\mapsto \psi_u(t)$ is compact, on appropriate spaces. 

The limitation of these statements is that
the spacial form of the external field  $V(x)$ is fixed, and only its magnitude is variable.  
In this note we show that this restriction can be removed for the Schrödinger equation.

This generalisation is based on the observation that the map assigning to a potential $V(x)$ the solution $\psi_V(t)$, where now $u$ is fixed, is also compact. 
While earlier results do not use any specific properties of the generator $A$, this extension relies on the local smoothing effect in the Schrödinger equation on $L^2(\R^d)$.
For the problem of controlling the equation, this means that a dense subset of the sphere in $L^2(\R^d)$ cannot be reached, even if one may not only choose $t$ and the function $u(t)$, but also $V(x)$, or more generally $V(t,x)$. 
More precisely, our main result (Theorem~\ref{thm1} below) implies that the Schrödinger equation on $L^2(\R^d)$ is not exactly controllable by potentials $V\in L^p_\mathrm{loc}(\R, L^\infty(\R^d))$ with $p>1$.

By contrast, there are some positive controllability results for Schrödinger equations on bounded domains in one and two dimensions, due to Beauchard, Coron~\cite{beauchard2006} and Beauchard, Laurent~\cite{beauchard2010, beauchard2016}. 
Under less restrictive assumptions only approximate controllability (i.e. density of the attainable set) is known, see e.g.~\cite{chambrion2009, mason2010, nersesyan2010}.
All of these positive results have in common that the generator $A$ of the uncontrolled dynamics has purely discrete spectrum. This precludes the dispersive local smoothing effects that our proof relies on (see the end the next section for a more technical remark on this point).
Whether approximate controllability still holds in our setting is a very interesting open question.

Our results also have implications for problems related to time-dependent density functional theory. There, one is interested in the map that assigns to a time-dependent potential $V(t,x)$ the density $\rho_V(t,x)=|\psi_V(t,x)|^2$ (or, in many particle systems, the one-particle density, a marginal of $|\psi_V|^2$). The Runge-Gross theorem~\cite{runge1984} states that this map is (essentially) one-to-one (though the proof requires very strong hypothesis on both $V$ and the initial conditon, see~\cite{tddft}). Our results will show that the range of this map, whose elements are called $V$-representable densities, has empty interior. This is an obstruction to constructing a local inverse, e.g.~by the inverse function theorem.

\section{Results}\label{sect:res}

As stated in the introduction, we would like to consider potentials in $L^p_\mathrm{loc}(\R, L^\infty(\R^d))$, $p>1$. However, to obtain local compactness, it will be convenient to work in spaces that are duals of Banach spaces, and $L^p(I,L^\infty(\R^d))$, $I\subset \R$ a compact interval, is \emph{not} the dual of $L^{p'}(I,L^1(\R^d))$, but strictly contained therein (except for $p=\infty$; see Diestel and Uhl~\cite{DieUhl}).
We will thus need to work in a somewhat larger space of \enquote{potentials}.
The elements of $\left(L^{p'}(I,L^1(\R^d))\right)'$ can naturally be identified with $L^\infty$-valued measures that are absolutely continuous with respect to the Lebesgue measure (c.f.~\cite[Thm.IV.1.1]{DieUhl}). Since $L^\infty$ does not have the Radon-Nikodym property (this can be seen by combining~\cite[Thm.III.3.2, Ex.III.1.1]{DieUhl}), there exist such absolutely continuous measures that do not have a density, and thus cannot be identified with a function.

\begin{definition}
 For a compact interval $I\subset \R$, $p\geq 1$ and $p'=(1-1/p)^{-1}$ we define
 \begin{equation}
  M^p(I):=\left(L^{p'}\left(I,L^1(\R^d,\R)\right)\right)'.
 \end{equation}
We define $M^p_\mathrm{loc}(\R)$ as those $L^\infty$-valued Borel measures on $\R$ whose restriction to every compact $I$ is an element of $M^p(I)$.
\end{definition}

For any $V\in M^p(I)$ the integral with respect to $V$ of a function $f\in L^\infty(I,L^2(\R^d))$ defines an element of $L^2(\R^d)$ by duality, i.e.,
\begin{equation}\label{eq:int def}
 \left\langle \int_I f(t) V(\ud t), g \right \rangle 
 := \left\langle V, f(t)g\right\rangle_{M^p \times L^{p'}(I,L^1(\R^d))} \qquad \forall g\in L^2(\R^d). 
\end{equation}
Since $f(t) V(\ud t)$ is an absolutely continuous $L^2$-valued measure and $L^2$, being reflexive, has the Radon-Nikodym property~\cite[Cor.III.13]{DieUhl}, there exists a function $\phi\in L^p(I,L^2(\R^d))$ such that $f(t) V(\ud t)=\phi(t) \ud t$.
Using this, 
\begin{equation}\label{eq:density}
 \int_I \ue^{-\ui\Delta t} f(t) V(\ud t)  = \int_I \ue^{-\ui\Delta t} \phi(t) \ud t
\end{equation}
is clearly well defined.
The integral form of the Schrödinger equation
\begin{equation}\label{eq:SE mild}
 \psi_V(t)=\ue^{\ui\Delta t} \psi_0 -\ui \int_0^t \ue^{\ui \Delta(t-s)} \psi_V(s)V(\ud s)
\end{equation}
thus makes sense for any $V\in M^p_\mathrm{loc}(\R)$.

\begin{lemma}
 Let $V\in M^p_\mathrm{loc}(\R)$ with $p\geq 1$, and $\psi_0\in L^2(\R^d)$. The equation~\eqref{eq:SE mild} has a unique solution $\psi_V\in C^0(\R,L^2(\R^d))$.
\end{lemma}
\begin{proof}
For small $t$ this follows from Banach's fixed point theorem.
For a solution $\psi_V$ to~\eqref{eq:SE mild}, $\ue^{-\ui \Delta t}\psi_V(t)$ is absolutely continuous with derivative in $L^p(I,L^2(\R^d))$.
This implies that $t\mapsto\|\psi_V(t)\|^2$ is absolutely continuous, from which one easily deduces that  $\|\psi_V(t)\|_2^2$ is constant by calculating its derivative.
This gives global existence and uniqueness by standard arguments.
\end{proof}

Now that we have solutions to~\eqref{eq:SE mild}, we define the attainable set by
\begin{equation}\label{eq:Ap}
 \mathcal{A}:=\bigcup_{p>1}
 \left\{ \psi_V(t) \Big|t \in \R,  \psi_V \text{ solves }\eqref{eq:SE mild} \text{ with } V\in M^p_\mathrm{loc}(\R) \right\}.
\end{equation}
This is naturally thought of as a subset of the sphere 
\begin{equation}
 \mathcal{S}_{\psi_0}:=\left\{ \psi \in L^2(\R^d) \Big\vert \|\psi\|_2=\|\psi_0\|_2 \right\}. 
\end{equation}
We also consider the set of $V$-representable densities on $[0,T]$, $T > 0$,
\begin{equation}
 \mathcal{R}_{T} :=\left\{ t \mapsto |\psi_V(t)|^2 \Big|  \psi_V \text{ solves }\eqref{eq:SE mild} \text{ on } [0,T] \text{ with } V\in M^p([0,T]) \right\},
\end{equation}
as a subset of 
\begin{equation}
 \mathcal{D}_{\psi_0}=\left\{ \rho\in C^0\left([0,T],L^1(\R^d)\right) \Big| \rho(0)=|\psi_0|^2, \, \int \rho(t,x) \ud x =\|\psi_0\|_2^2 \right\}.
\end{equation}

\begin{thm}\label{thm1}
For every $\psi_0\in L^2(\R^d)$ and $T> 0$ the sets $\mathcal{A}$ 
and $\mathcal{R}_{T}$ are countable unions of compact subsets of $L^2(\R^d)$, respectively $C^0\left([0,T],L^1(\R^d)\right)$.
\end{thm}

This theorem will be proved in Section~\ref{sect:proof} below.

\begin{corol}\label{cor:A}
The set $\mathcal{A}$  has empty interior in $\mathcal{S}_{\psi_0}$.
\end{corol}
\begin{proof}
From Theorem~\ref{thm1} we know that $\mathcal{A}$ is contained in a countable union of compact sets. Baire's theorem states that, in a complete metric space, the countable union of closed sets with empty interior also has empty interior, c.f.~Brezis~\cite[Thm. 2.1]{Brezis}. Since $\mathcal{S}_{\psi_0}$ is complete, we thus only need to prove that any compact set
  $K\subset \mathcal{S}_{\psi_0}$ has empty interior.
 Assume to the contrary that there is such a $K$ with non-empty interior. Then consider the truncated cone
 \begin{equation}
  C= \bigcup_{0\leq r\leq 1} r K.
 \end{equation}
 One easily checks that $C\subset L^2(\R^d)$ is compact with non-empty interior, since $K$ has these properties. But this is impossible, since a compact set in $L^2(\R^d)$ cannot contain an open set, due to Riesz' theorem~\cite[Thm. 6.5]{Brezis}.
\end{proof}

\begin{corol}
  For any $T> 0$ the set $\mathcal{R}_{T}$ has empty interior in $\mathcal{D}_{\psi_0}$.
\end{corol}
\begin{proof}
As the conditions $\rho(0)=|\psi_0|^2$ and $\int \rho(t,x) \ud x = \|\psi_0\|_2^2$ characterising $\mathcal{D}_{\psi_0}$ are linear, this follows immediately from Theorem~\ref{thm1} and the theorems of Baire and Riesz. 
\end{proof}
The same result holds if we replace $\mathcal{R}_{T}$ by the set obtained by taking a marginal of $|\psi_V |^2$, as will be clear from the proof of Theorem~\ref{thm1}.

\begin{rem*}
Let us make some remarks on possible generalisations of our results. 
\begin{enumerate}
 \item Unbounded potentials: Our conclusions also hold for potentials that are locally in 
\begin{equation}
 \left(L^{p'}\left(I,L^{q'}(\R^d,\R)\right)\right)' \cong L^{p}\left(I,L^{q}(\R^d,\R) \right) \qquad q<\infty,
\end{equation}
if $q\geq 2$, $q>d$ and $p>2q/(2q-d)$. In this case, one can circumvent the use of the spaces of measures $M^p$, but since such functions no longer take their values in the bounded operators on $L^2(\R^d)$ the proofs require some finer estimates. These can be obtained using ideas of Frank, Lewin, Lieb, Seiringer~\cite{FLLS}, in particular the generalised Kato-Seiler-Simon inequality~\cite[Lem. 1]{FLLS} (which together with the Hardy-Lttlewood-Sobolev inequality implies the analogue of Lemma~\ref{lem:L cont}, for example). The structure of the argument remains largely unchanged.

\item Sobolev spaces: Since the local smoothing property, the main ingredient of our proof, holds on the scale of Sobolev spaces,
one can consider classes of more regular potentials and initial conditions and prove that the corresponding attainable set has empty interior in $H^s(\R^d)$, $s>0$.

\item Manifolds: Our ideas should also work for the Schrödinger equation on Riemannian manifolds, provided that $\ue^{\ui t \Delta_g}$, where $\Delta_g$ is the Laplace-Beltrami operator on $(M,g)$, has an appropriate local smoothing property. This is closely related to the properties of the geodesic flow~\cite{rodnianski2015}. Roughly speaking, the geodesic flow should be  non-trapping.
Notably, this excludes compact manifolds (with boundary), where exact controllability is known to hold in some cases~\cite{beauchard2006,beauchard2010, beauchard2016}) (these results are also restricted to very regular functions, but regularity alone will not be enough to obtain controllability, by the previous remark).
\end{enumerate}
\end{rem*}

\section{Proof of Theorem~\ref{thm1}}\label{sect:proof}

To study the map $V \mapsto \psi_V$ we will first  look at its linearisation.
Let $I\subset \R$ be a compact interval and $f,g\in L^\infty(I,L^2(\R^d))$. We define a linear operator 
\begin{equation}
L_I(f,g): M^p(I) \to L^{p'}(I,L^1(\R^d))\label{eq:L_I def}
\end{equation}
by setting for any $W\in M^p(I)$
\begin{align}
  \big\langle W&, L_I(f,g) V \big\rangle_{M^p\times L^{p'}(I,L^1(\R^d))} \notag\\
  :&= \left\langle \int_I \ue^{-\ui\Delta s} g(s) W(\ud s), \int_I \ue^{-\ui\Delta s'} f(s') V(\ud s')\right\rangle_{L^2}\notag\\
  &= \left\langle W, \overline{g(s)} \int_I \ue^{\ui\Delta(s- s')} f(s') V(\ud s')\right \rangle_{M^p\times L^{p'}(I,L^1(\R^d))}.
  \label{eq:L_I form}
\end{align}

\begin{lemma}\label{lem:L cont}
 For all $I,f,g$ as above, the operator $L_I(f,g)$ is continuous. It depends continuously on $f$ and $g$.
\end{lemma}
\begin{proof}
 Since $L_I$ is linear in all of the arguments $f,g,V$, we only need to show that it is locally bounded. From the definition of the integral~\eqref{eq:int def}, we immediately obtain that
 \begin{equation}
  \left\| \int_I \ue^{\ui\Delta(s- s')} f(s') V(\ud s')\right\|_{L^2(\R^d)} \leq \|V\|_{M^p(I)} \|f\|_{L^{p'}(I,L^2(\R^d))}.
 \end{equation}
Using the Hölder inequality $\|f\|_{L^{p'}(I)}\leq |I|^{1-1/p} \|f\|_{\infty}$, this implies the bound 
\begin{equation}\label{eq:L_I bound}
 \|L_I(f,g) V\|_{L^{p'}(I,L^1(\R^d))} \leq |I|^{2-2/p}  \|g\|_{L^\infty(I, L^2(\R^d))}\|f\|_{L^\infty(I,L^2(\R^d))} \|V\|_{M^p}.
\end{equation}
\end{proof}

The main technical lemma is to show that $L_I$ is compact.

\begin{lemma}\label{lem:L comp}
 Let $1<p\leq \infty$, $I\subset \R$ compact and $f,g\in C^0(I,L^2(\R^d))$. The operator $L_I(f,g)$ given by~\eqref{eq:L_I def},~\eqref{eq:L_I form} is compact.
\end{lemma}
\begin{proof}
 Since the set of compact operators is norm-closed, it is sufficient to prove the claim for a sequence of approximations of $L_I(f,g)$. By the continuous dependence of $L_I(f,g)$ on $f$ and $g$ we may thus assume that both of these functions take their values in some bounded subset of the Schwartz space $\mathscr{S}(\R^d)$.
 
 We now make this assumption of $f,g$ and suppress them in the notation for $L_I$.
 In order to prove that images of bounded subsets are pre-compact, 
 We first note that the image of $L_I$ is contained in $C^0(I,L^1(\R^d))$, since $V\in M^p$ is absolutely continuous. We will thus be able to prove compactness by applying the vector-valued Arzelà-Ascoli theorem (see Lang~\cite[Thm. III.3.1]{LangFA}), as the inclusion of $C^0$ into $L^p$ is continuous.
 
 We start by proving that the point-wise image is pre-compact, that is, for every $s\in I$ and $R>0$ the set
 \begin{equation}\label{eq:s-comp}
  \left\{ (L_I V)(s) \Big| \|V\|_{M^p}<R \right\}
 \end{equation}
is pre-compact in $L^1(\R^d)$.
From the formula~\eqref{eq:L_I form} for $L_I$ we immediately see that functions in the range of $L_I$ are rapidly decreasing since $g(s)$ is, i.e. for any $r\geq 0$
\begin{equation}
  \||x|^r (L_I V)(s,x) \|_{1} 
  \leq C_I \| |x|^r g(s) \|_{2}  \left\|f\right\|_{L^\infty(I, L^2(\R^d))} \|V\|_{M^p}.
  \label{eq:r-decay}
\end{equation}
Let $\phi\in L^p(I,L^2(\R^d))$ be the density of $f(s)V(\ud s) = \phi(s) \ud s$.
To obtain a bound on derivatives of $L_I V$, we use that for $j=1,\dots,d$ 
\begin{equation}
-2\ui \tau \partial_j \ue^{\ui \Delta \tau}= \left[x_j, \ue^{\ui \tau \Delta}\right]
\end{equation}
to calculate
\begin{align}
& \overline{g(s)} \partial_j^2 \ue^{\ui \Delta\tau} \phi(s')\notag 
 = \overline{g(s)}\frac{2\ui \tau \partial_j}{4\tau^2} \left[x_j, \ue^{\ui \tau \Delta}\right]\phi(s')\notag \\
 %
 &=-\overline{g(s)}\left(\ue^{\ui \Delta\tau} \frac{x^2_j}{4\tau^2} \phi(s') +  x_j \ue^{\ui \Delta\tau} \frac{x_j}{2\tau^2} \phi(s') + \left(\frac{\ui}{2\tau}+ \frac{x^2_j}{4\tau^2} \right)\ue^{\ui \Delta\tau} \phi(s')\right).\label{eq:Lapl comm}
 \end{align}
 Using the regularity of $g$ and that $s,s'$ belong to the bounded interval $I$, we obtain from this the bound
\begin{equation}
 \left\|\overline{g(s)} \ue^{\ui \Delta(s-s')} \phi(s')\right\|_{W^{2,1}}
 \leq \frac{C \left(\|(1+x^2)g(s)\|_2 + \|\Delta g(s)\|_2 \right)}{|s-s'|^2} \|(1+x^2)\phi(s')\|_2,
\end{equation}
and, by interpolation for $0\leq m\leq 2$,
\begin{equation}
 \left\|\overline{g(s)} \ue^{\ui \Delta(s-s')} \phi(s')\right\|_{W^{m,1}}\leq \frac{C\left(\|(1+x^2)g(s)\|_2 + \|\Delta g(s)\|_2 \right)}{|s-s'|^m} \|(1+x^2)\phi(s')\|_2.
\end{equation}
Since $x^2 f(s,x)\in L^\infty(I,L^2(\R^d))$, we have for all 
$m$ such that $|s|^{-m}\in L^{p'}(I)$, i.e. 
$m < 1-1/p$,
\begin{equation}
  \|(L_I V)(s,x) \|_{W^{m,1}} \leq C(m,f,g) \|V\|_{M^p}.
  \label{eq:m-deriv}
\end{equation}
The set of functions in $L^1(\R^d)$ satisfying the bounds~\eqref{eq:r-decay} and~\eqref{eq:m-deriv} with $\|V\|_{M^p} < R$ is pre-compact as a corollary of the Rellich-Kondrachov theorem~\cite[Thm. 9.16]{Brezis} (see~\cite[Cor. 4.27]{Brezis}), which gives the pre-compactness of~\eqref{eq:s-comp}.

In order to apply the Arzelà-Ascoli theorem it remains to show that images of bounded sets are uniformly equi-continuous. Since $f,g$ are continuous functions on the compact interval $I$ they are equi-continuous, so $f(s+h)-f(s)$ and $g(s+h)-g(s)$ converge to zero in $L^\infty(I,L^2(\R^d))$ as $h\to 0$. In view of Lemma~\ref{lem:L cont} it is thus sufficient to show that for every $\eps>0$ there exists $\delta>0$ such that for all $h<\delta$
\begin{align}
 &\left\| L_I(f,g)V -\overline{g(s)} \int_I \ue^{\ui \Delta(s+h-s')} f(s') V(\ud s') \right\|_{L^{p'}(I,L^1(\R^d))}
 %
 \leq \eps \|V\|_{M^p}.\label{eq:L equi}
\end{align}
Using the formula~\eqref{eq:Lapl comm} and the fact that the commutator $[(1-\Delta)^{-1},x]$ is a bounded operator from $L^2(\R^d)$ to $H^2(\R^d)$, we obtain
\begin{align}
& \left \| \overline{g(s)} (\ue^{\ui\Delta h} -1 )  \ue^{\ui \Delta(s-s')} \phi(s')\right\|_1 \notag\\
& \leq \frac{C}{|s-s'|^2} \| \overline{g(s)} (\ue^{\ui\Delta h} -1 )(1-\Delta)^{-1}(1+x^2)\|_{L^2 \to L^1} \|(1+x^2)\phi(s')\|_2 \notag\\
&\leq \frac{h}{|s-s'|^2} C(g) \|(1+x^2)\phi(s')\|_2.
\end{align}
Interpolating between this and the trivial estimate, obtained using $\|1-\ue^{\ui\Delta h}\|\leq 2$, yields
\begin{equation}
 \left \| \overline{g(s)} (\ue^{\ui\Delta h} -1 )  \ue^{\ui \Delta(s-s')} \phi(s')\right\|_1 \leq \frac{h^{m/2}}{|s-s'|^m} C(g)\|(1+x^2)\phi(s')\|_2.
\end{equation}
Taking $m< 1-1/p$ and using Young's inequality then gives the bound
\begin{equation}
 \left\| \overline{g(s)} \int_I (1-\ue^{\ui \Delta h})\ue^{\ui \Delta(s-s')} \phi(s') \ud s' \right\|_{L^{p'}(I,L^1(\R^d))}
 \leq C(f,g)  h^{m/2} \|V\|_{M^p},
\end{equation}
which proves~\eqref{eq:L equi}.

Having established~\eqref{eq:s-comp} and~\eqref{eq:L equi}, we can now apply the Arzelà-Ascoli theorem~\cite[Thm. III.3.1]{LangFA}, and this proves the claim. 
\end{proof}

The compactness of $L_I$, the square of the linearised solution operator, can be lifted to the full solution operator $V\mapsto \psi_V$.

\begin{prop}\label{prop:unif}
 Let $p>1$, $I$ a compact interval and let $V_n$, $n\in \N$ be a bounded sequence in $M^p(I)$ that converges weakly-$*$ to $V\in M^p(I)$.
 Then the corresponding solutions $\psi_{V_n}$ converge to $\psi_V$ in the norm of $C^0(I,L^2(\R^d))$.
\end{prop}
\begin{proof}
 Let $D_n(t):=\sup_{0\leq \tau \leq t} \|\psi_{V_n}(\tau) - \psi_{V}(\tau) \|_2$. Then, writing simply $L_t$ for the operator $L_{[0,t]}(\psi_V,\psi_V)$, we have
 \begin{equation}
 D_n(t) \leq \sup_{0\leq \tau \leq t}
  \left( \sqrt{\langle V_n-V, L_{\tau}(V_n-V) \rangle} + \left\|\int_0^\tau\ue^{-\ui\Delta s}  (\psi_{V}(s)-\psi_{V_n}(s)) V_n(\ud s) \right\|_2\right).
 \end{equation}
By the Radon-Nikodym property of $L^2(\R^d)$ there exists $\phi_n \in L^p([0,T],L^2(\R^d))$ such that 
\begin{equation}
 \frac{\psi_{V}(s)-\psi_{V_n}(s)}{\|\psi_{V}(s)-\psi_{V_n}(s)\|_2 } V_n( \ud s) = \phi_n(s) \ud s,
\end{equation}
and $\int \|\phi_n \|_2^p \leq \|V_n\|^p_{M^p}$.
Inserting this into the previous inequality, we obtain
\begin{equation}
 D_n(t) \leq \sup_{0\leq \tau \leq t} \sqrt{\langle V_n-V, L_{\tau}(V_n-V) \rangle} +  \int_0^t D_n(s) \|\phi_n(s) \|_2 \ud s.
\end{equation}
Grönwall's inequality then gives
\begin{equation}
 D_n(T) \leq \left(\sup_{0\leq t \leq T} \sqrt{\langle V_n-V, L_{t}(V_n-V) \rangle}\right)
 \exp \left( \int_0^T \|\phi_n(s)\|_2 \ud s \right).
\end{equation}
The proof will thus be complete if we can show that the first factor tends to zero.
We already know, from Lemma~\ref{lem:L comp}, that $\langle V_n-V, L_{t}(V_n-V)\rangle$ converges to zero point-wise. To show that this convergence is uniform, we will prove that the sequence is uniformly equi-continuous on $[0,T]$. To do this, let $t'<t$ and extend functions in $L^{p'}([0,t'])$ by zero to obtain functions in $L^{p'}([0,t])$.
With this convention, we  have
\begin{equation}
 L_{t}W - L_{t'}W = \chi_{[t',t]} L_t W +  \chi_{[0,t']} L_{[t',t]} W,
\end{equation}
where $\chi_I$ denotes the characteristic function of $I$.
The bound~\eqref{eq:L_I bound} then implies
\begin{equation}
 |\left\langle V_n-V, L_{t}V_n-V \right\rangle - \left\langle V_n-V, L_{t'}V_n-V \right\rangle| \leq C \|V_n-V\|^2_{M^p} |t-t'|^{1-1/p}.
\end{equation}
Since the sequence $V_n$ was assumed to be bounded, this gives the required uniform continuity, and thus proves our claim.
\end{proof}

It is now rather straightforward to prove Theorem~\ref{thm1}.

\begin{lemma}\label{lem:A T,p,R}
Let $T\geq 0$, $R\geq 0$ and $1<p\leq \infty$. The set 
\begin{equation}
\mathcal{A}_{T,p,R} = \left\{ \psi_V(t) \Big| 0\leq t\leq T,\, \|V\|_{M^p([0,T])} \leq R \right\} 
\end{equation}
is compact in $L^2(\R^d)$.
\end{lemma}
\begin{proof}
 Let $\psi_{V_n}(t_n)$, $n\in \N$ be a sequence in $\mathcal{A}_{T,p,R}$. 
 The space $M^p([0,T])$ is locally weak-$*$ compact by the Banach Alaoglu theorem. Since $L^{p'}([0,T],L^1(\R^d))$, $p'=(1-1/p)^{-1}<\infty$, is separable, the weak-$*$ topology is locally metrisable.
 We can thus extract a subsequence, denoted by the same symbols, such that $t_{n}\to t\in [0,T]$ and $V_{n}\rightharpoonup^* V\in M^p([0,T])$ as $n\to \infty$. Then $\psi_{V_{n}} \to \psi_{V}$ uniformly by Proposition~\ref{prop:unif}, which implies convergence of $\psi_{V_{n}}(t_n)$ to $\psi_V(t)$. 
 This shows that $\mathcal{A}_{T,p,R}$ is compact.
\end{proof}

Using Lemma~\ref{lem:A T,p,R} we can write $\mathcal{A}$ as a countable union of compact sets:
The inclusion $M^p([0,T]) \subset M^q([0,T])$, $p\geq q$ gives $\mathcal{A}_{T,p,R} \subset \mathcal{A}_{T,q,R}$, and obviously $\mathcal{A}_{T,p,R} \subset \mathcal{A}_{T',p,R'}$ for $T\leq T'$, $R\leq R'$. Hence for a decreasing sequence $p_n>1$, $n\in \N$, converging to one, the sets $\mathcal{A}_{n,p_n,n}$ form an increasing sequence of compact sets and
\begin{equation}\label{eq:A union}
 \mathcal{A}= \bigcup_{n\in \N} \mathcal{A}_{n,p_n,n}.
\end{equation}

Essentially the same argument holds for $\mathcal{R}_T$ in view of the following lemma.

\begin{lemma}
 Let $T\geq 0$, $R\geq 0$ and $1<p\leq \infty$. The set 
\begin{equation}
  \mathcal{R}_{T,p,R}:=\left \{ t\mapsto |\psi_V(t)|^2 \Big| \|V\|_{M^p([0,T])} \leq R \right\}
\end{equation}
is compact in $C^0([0,T],L^1(\R^d))$.
\end{lemma}
\begin{proof}
 By Proposition~\ref{prop:unif} and the local weak-$*$ compactness of $M^p([0,T])$, the set of trajectories $t\mapsto \psi_V(t)$ is compact in $C^0([0,T],L^2(\R^d))$. This implies the claim, since the map $\psi\mapsto |\psi|^2$ is continuous. 
\end{proof}
Note that this statement stays exactly the same if we take additionally a marginal of $|\psi_V(t)|^2$.
We then have
\begin{equation}
\mathcal{R}_T = \bigcup_{n\in \N} \mathcal{R}_{T,p_n,n},
\end{equation}
with a sequence $p_n\to 1$ as in~\eqref{eq:A union}. This completes the proof of Theorem~\ref{thm1}.

\section*{Acknowledgements}

The author wishes to thank Nabile Boussa\"{\i}d for explaining to him the results of~\cite{ball1982, boussaid2014, boussaid2017}. He also thanks Mathieu Lewin for many discussions on related topics. Financial support from the ANR QUACO (PRC ANR-17-CE40-0007-01) and the 
 European Research Council (ERC) under the European Union’s Horizon 2020 research and innovation programme
 (grant no.725528 MDFT, principal investigator M. Lewin) is gratefully acknowledged.

%

\end{document}